\documentclass[10pt,twocolumn,twoside,aps,pra,nopacs,superscriptaddress,a4paper]{revtex4-1}
\usepackage{graphicx}
\usepackage{amsmath,amssymb}
\usepackage{color}
\usepackage{bbm}
\usepackage{hyperref}
\usepackage{theorem}
\usepackage{latexsym}
\usepackage{enumitem}

\newtheorem{theorem}{Theorem}

\newtheorem{corollary}[theorem]{Corollary}

\newtheorem{definition}[theorem]{Definition}
\newtheorem{example}[theorem]{Example}

\newtheorem{lemma}[theorem]{Lemma}

\newtheorem{proposition}[theorem]{Proposition}
\newtheorem{remark}[theorem]{Remark}

\newlength{\blank}
\newenvironment{proof}[1][{\hspace{-\blank}}]{{\medskip\noindent\textbf{Proof~{#1}.\ }}}{\hfill\qed}

\newcommand{\ket}[1]{|#1\rangle}
\newcommand{\bra}[1]{\langle#1|}

\mathchardef\ordinarycolon\mathcode`\:
\mathcode`\:=\string"8000
\def\vcentcolon{\mathrel{\mathop\ordinarycolon}}
\begingroup \catcode`\:=\active
  \lowercase{\endgroup
  \let :\vcentcolon
  }

\newcommand{\nc}{\newcommand}
\nc{\rnc}{\renewcommand}
\nc{\lbar}[1]{\overline{#1}}
\nc{\ketbra}[2]{|#1\rangle\!\langle#2|}
\nc{\proj}[1]{| #1\rangle\!\langle #1 |}
\nc{\avg}[1]{\langle#1\rangle}
\nc{\Rank}{\operatorname{Rank}}
\nc{\smfrac}[2]{\mbox{$\frac{#1}{#2}$}}
\nc{\tr}{\operatorname{Tr}}
\nc{\ox}{\otimes}
\nc{\dg}{\dagger}
\nc{\dn}{\downarrow}
\nc{\cA}{\mathcal{A}}
\nc{\cB}{\mathcal{B}}
\nc{\cC}{\mathcal{C}}
\nc{\cD}{\mathcal{D}}
\nc{\cE}{\mathcal{E}}
\nc{\cF}{\mathcal{F}}
\nc{\cG}{\mathcal{G}}
\nc{\cH}{\mathcal{H}}
\nc{\cI}{\mathcal{I}}
\nc{\cJ}{\mathcal{J}}
\nc{\cK}{\mathcal{K}}
\nc{\cL}{\mathcal{L}}
\nc{\cM}{\mathcal{M}}
\nc{\cN}{\mathcal{N}}
\nc{\cO}{\mathcal{O}}
\nc{\cP}{\mathcal{P}}
\nc{\cR}{\mathcal{R}}
\nc{\cS}{\mathcal{S}}
\nc{\cT}{\mathcal{T}}
\nc{\cX}{\mathcal{X}}
\nc{\cZ}{\mathcal{Z}}
\nc{\csupp}{{\operatorname{csupp}}}
\nc{\qsupp}{{\operatorname{qsupp}}}
\nc{\var}{\operatorname{var}}
\nc{\rar}{\rightarrow}
\nc{\lrar}{\longrightarrow}
\nc{\polylog}{\operatorname{polylog}}
\nc{\1}{{\openone}}
\nc{\id}{{\operatorname{id}}}

\nc{\RR}{{{\mathbb R}}}
\nc{\CC}{{{\mathbb C}}}
\nc{\FF}{{{\mathbb F}}}
\nc{\NN}{{{\mathbb N}}}
\nc{\ZZ}{{{\mathbb Z}}}
\nc{\PP}{{{\mathbb P}}}
\nc{\QQ}{{{\mathbb Q}}}
\nc{\UU}{{{\mathbb U}}}
\nc{\EE}{{{\mathbb E}}}

\nc{\GOCC}{{\mathrm{GOCC}}}
\nc{\Wplus}{{\mathrm{W+}}}

\nc{\qed}{{$\hfill\Box$}}

\begin{document}

\title{Bosonic data hiding: power of linear vs non-linear optics}

\author{Krishna Kumar Sabapathy}
\email{krishnakumar.sabapathy@gmail.com}
\affiliation{Xanadu, 777 Bay Street, Toronto ON, M5G 2C8, Canada}
\affiliation{Departament de F\'{\i}sica: Grup d'Informaci\'{o} Qu\`{a}ntica, %
Universitat Aut\`{o}noma de Barcelona, 08193 Bellaterra (Barcelona), Spain}

\author{Andreas Winter}
\email{andreas.winter@uab.cat}
\affiliation{Departament de F\'{\i}sica: Grup d'Informaci\'{o} Qu\`{a}ntica, %
Universitat Aut\`{o}noma de Barcelona, 08193 Bellaterra (Barcelona), Spain}
\affiliation{ICREA---Instituci\'o Catalana de Recerca i Estudis Avan\c{c}ats, %
Pg.~Lluis Companys, 23, 08010 Barcelona, Spain} 

\date{2 February 2021}

\begin{abstract}
  We show that the positivity of the Wigner function of Gaussian states
  and measurements provides an elegant way to bound the discriminating
  power of ``linear optics'', which we formalise as 
  Gaussian measurement operations augmented by classical (feed-forward) 
  communication (GOCC). 
  This allows us to reproduce and generalise the result of 
  Takeoka and Sasaki [PRA 78:022320, 2008],
  which tightly characterises the GOCC norm distance of coherent states,
  separating it from the optimal distinguishability according to 
  Helstrom's theorem.
  
  Furthermore, invoking ideas from classical and quantum Shannon theory
  we show that there are states, each a probabilistic mixture of multi-mode
  coherent states, which are exponentially reliably discriminated in
  principle, but appear exponentially close judging from the output
  of GOCC measurements. In analogy to LOCC data hiding,
  which shows an irreversibility in the preparation and discrimination of
  states by the restricted class of local operations and classical communication
  (LOCC), we call the present effect \emph{GOCC data hiding}.
  
  We also present general bounds in the opposite direction, guaranteeing
  a minimum of distinguishability under measurements with positive Wigner 
  function, for any bounded-energy states that are Helstrom distinguishable.
  We conjecture that a similar bound holds for GOCC measurements. 
\end{abstract}


\maketitle

\section{Introduction}
\label{sec:intro}
One of the most basic problems of quantum information theory is the 
discrimination of two alternatives (``hypotheses''), each of which 
represents the possible state 
of a system, $\rho_0$ or $\rho_1$. Under the formalism of quantum mechanics, this 
calls for the design of a measurement and a decision rule to choose between the
two options based on the measurement outcome. The measurement is a
binary resolution of unity, also called a positive operator valued measure 
(POVM), $(M_0=M,M_1=\1-M)$ of two semidefinite operators $M_0\,M_1 \geq 0$
summing to $M_0+M_1=\1$. The outcome $M_{\hat{i}}$ of the measurement is intended 
to correspond to the estimate $\hat{i}$ of the true state $\rho_i$.
For simplicity, we will assume that the two hypotheses come with equal (uniform) 
prior probabilities, so the error probability is 
\begin{equation}\begin{split}
  \label{eq:HelstromHolevo}
  P_e &= \frac12 \tr\rho_0 (\1-M) + \frac12 \tr\rho_1 M \\
      &= \frac12 \bigl( 1-\tr (\rho_0-\rho_1)M \bigr).
\end{split}\end{equation}
The minimum error over all quantum mechanically allowed POVMs gives 
rise to the trace norm,
\begin{equation}
  \label{eq:tracenorm}
  \min_{0\leq M\leq\1} P_e = \frac12 \left( 1-\frac12\|\rho_0-\rho_1\|_1 \right),
\end{equation}
which formula is nowadays known as Helstrom bound \cite{Helstrom,Helstrom:book}
or Holevo-Helstrom bound \cite{Holevo:dist}, since it was initially only proved for 
projective measurements and subsequently for generalised measurements. 

However, from the beginning of quantum detection theory, it was understood 
that -- depending on the physical system -- the Helstrom optimal measurement 
may not be easily implemented. Indeed, the very example of discrimination of 
to coherent states of an optical mode was already considered by Helstrom \cite{Helstrom:coherent}, 
who contrasted the absolutely minimum error probability with the performance 
of reasonable practical measurements. 
Mathematically, this means that the minimisation on the l.h.s. of 
Eq.~(\ref{eq:tracenorm}) is performed over a smaller set of POVMs,
the restriction an expression of what is deemed physically feasible.
Consequently, the error probability becomes larger, in some interesting 
cases close to $\frac12$ even for orthogonal, i.e.~ideally perfectly 
distinguishable, states. This phenomenon was first observed in bipartite 
systems under the restriction of local operations and classical communication
(LOCC), and dubbed \emph{data hiding} \cite{Terhal-datahiding,DiVincenzo-datahiding}, 
which has been generalised to multi-party settings \cite{EggelingWerner},
and analysed extensively \cite{MWW,LW,W:eff}.

In the present paper we will look at a different kind of restriction, 
in Bosonic quantum systems, motivated by the distinction between 
phase-space \emph{linear} (aka \emph{Gaussian}) and \emph{non-linear}
(i.e.~non-Gaussian) operations, see also \cite{KKVV}. 
It is well-known that a process that 
starts from a Gaussian state and proceeds only via Gaussian operations,
including Gaussian measurements and classical feed-forward (GOCC, see below),
is in a certain sense very far away from the full complexity of quantum 
mechanics: indeed, such a process can be simulated efficiently on a  
classical computer \cite{bartlett,mari} and hence, unless BQP=BPP, is not quantum computationally 
universal. In other words, non-Gaussianity is a resource for computation, 
which it becomes quite explicitly in proposals of optical quantum computing
such as the Knill-Laflamme-Milburn scheme \cite{KLM} that relies on 
photon detection and otherwise passive linear optics. 

Here we show that non-Gaussianity is a resource for the basic task of binary 
hypothesis testing. In particular we show how to leverage simple properties
of the Wigner function to prove not only a limitation of the power of 
Gaussian operations, but construct data hiding with respect to GOCC.

To conclude the introduction, a word on terminology: we refer to Gaussian states 
and channels as ``linear'', because the latter are described by linear 
transformations in the phase space of the canonical variables $x$ and $p$.
Conversely, ``non-linear'' is anything outside the Gaussian set. 
Note however that in parts of quantum optics a narrower concept is used, 
whereby only channels are considered linear that are built with passive 
Gaussian unitaries, and perhaps admitting displacement operators. 

\medskip
The rest of paper is structured as follows: In the next section (\ref{sec:gaussian})
we recall the necessary formalism and notation of quantum harmonic
oscillators and Gaussian Bosonic states and operations; for our purposes
in particular useful will be the phase space methods based on Wigner functions.
Then, in Section~\ref{sec:GOCC} we specialise the general framework
of restricted measurements to Gaussian quantum operations and arbitrary classical
computations (GOCC), and the important relaxation of this class to measurements
with non-negative Wigner functions (W+). 
We use these in Section~\ref{sec:GOCC-vs-ALL} to analyze the optimal GOCC
measurement to distinguish two coherent states, reproducing (with a conceptually
much simpler proof) a result of Takeoka and Sasaki~\cite{TakeokaSasaki}.
The GOCC distinguishability of any two distinct coherent states is always a 
little, but always strictly worse than the optimal distingishability 
according to Helstrom \cite{Helstrom,Holevo:dist}. 
Motivate by this, in Section~\ref{sec:hiding} we exhibit examples of multimode
states, each a mixture of coherent states (hence ``classical'' in the quantum-optical sense \cite{Glauber,Sudarshan}
and in particular preparable by GOCC), whose GOCC distinguishability is 
exponentially small while they are almost perfectly distinguishable under the 
optimal Holevo-Helstrom measurement. 
From the other side, there are lower limits to how indistinguishable two 
orthogonal states on $n$ quantum harmonic modes and with bounded energy can be,
which we show for W+ measurements and conjecture for GOCC measurements
(Section~\ref{sec:lower-bounds}).
We conclude in Section~\ref{sec:conclusions}.

\section{Bosonic Gaussian formalism}
\label{sec:gaussian}
We briefly review the formalism of Bosonic systems and Gaussian states, 
which has been laid out in many review articles and textbooks, such 
as \cite{Weedbrook-et-al} and \cite{Barnett,KokLovett}, which two emphasise the 
quantum information aspect. 
For our particular choice of normalisations, see \cite{cahill}. 

Each elementary system, called a (harmonic) mode, is characterised by 
a pair of canonical variables $x$ and $p$, satisfying the canonical 
commutation relation (CCR) $[x,p]=i$ (customarily choosing units where $\hbar=1$) 
and generating the CCR algebra of Heisenberg and Weyl. 
By the Stone-von-Neumann theorem, each irreducible representation of 
this algebra on a separable Hilbert space $\cH$ is isomorphic to the 
usual position and momentum operators $x$ and $p$, respectively. 

It is convenient to introduce the annihilation and creation operators
\begin{equation}
  \label{eq:anni+crea}
  a = \frac{x+ip}{\sqrt{2}}, \quad a^\dagger = \frac{x-ip}{\sqrt{2}}, 
\end{equation}
respectively. They can be used to define the number operator,
\begin{equation}
  \label{eq:N}
  N = a^\dagger a = \frac12 (x^2+p^2) - \frac12 = \sum_{n=0}^\infty n\,\proj{n}, 
\end{equation}
which up to the energy shift of $-\frac12$ to bring the ground state energy 
to zero, is equivalent to the quantum harmonic Hamiltonian (at fixed
frequency), and has precisely the non-negative integers as eigenvalues; 
the eigenstates are known as Fock states or number states, $\ket{n}$. 
In the number basis,
\begin{equation}
  \label{eq:anni+crea:N}
  a         = \sum_{n=0}^\infty \sqrt{n}\,\ketbra{n\!-\!1}{n}, 
  \quad 
  a^\dagger = \sum_{n=0}^\infty \sqrt{n}\,\ketbra{n}{n\!-\!1}. 
\end{equation}
All these operators are unbounded, and one might have justified 
hesitations against the algebraic operations performed above. The 
established solution to all of the potential problems associated 
to the unboundedness and associated restricted domains is to pass 
to the displacement operators,
\begin{equation}
  \label{eq:D-alpha}
  D(\alpha) = e^{\alpha a^\dagger-\overline{\alpha} a}, \text{ for } \alpha\in \CC,
\end{equation}
which are \emph{bona fide} unitaries, hence bounded operators. 

So far, we have discussed our quantum system at hand as if it were a single 
mode, but we can  of course consider multi-mode systems, which again by the
Stone-von-Neumann theorem are characterised uniquely as irreducible 
representations of the CCR algebra generated by $x_1,\ldots,x_m$ and
$p_1,\ldots,p_m$ such that $[x_j,p_k] = i\delta_{jk}$. This means that its
Hilbert space can be identified with $\cH_1\ox\cdots\ox\cH_m$, where
$\cH_j$ is the Hilbert space of the $j$-th mode, carrying the representation 
of $x_j$ and $p_j$. In particular, each mode has its own annihilation 
operator $a_j$ and displacement operator $D(\alpha)$; for an $m$-tuple 
$\underline{\alpha} = (\alpha_1,\ldots,\alpha_m)$ of displacements,
we write $D(\underline{\alpha}) = D(\alpha_1)\ox\cdots\ox D(\alpha_m)$ 
for the $m$-mode displacement operator. The subspace spanned by these
operators is dense in the bounded operators $\cB(\cH)$ on the Hilbert space. 

For a general density operator $\rho \in \cS(\cH) = \{\rho\geq 0,\ \tr\rho=1\}$, 
or more generally for a trace class operator, the characteristic 
function is defined as
\begin{equation}
  \chi_\rho(\underline{\alpha}) := \tr \rho D(\underline{\alpha}).
\end{equation}
This is a bounded complex function, uniquely specifying $\rho$. A state 
$\rho$ is called Gaussian if its characteristic function $\chi_\rho$ is 
of Gaussian form. 
For our purposes, we need another, particularly useful representation of the 
state as a quasi-probability function, the so-called Wigner 
function $W_\rho$ \cite{Wigner,HOCSW}, see also \cite{Barnett,cahill} for many 
fundamental and useful relations such as the following two. 
It is defined as the (multidimensional complex) Fourier transform of the 
characteristic function $\chi_\rho$, 
\begin{equation}
  \label{eq:Wigner-Fourier}
  W_\rho(x,p) 
    = \left(\frac{1}{2\pi^2}\right)^m 
      \int {\rm d}^{2m}\underline{\xi}\,
                  e^{\underline{\alpha}\cdot\underline{\xi}^\dagger
                     -\underline{\xi}\cdot\underline{\alpha}^\dagger}
                  \chi_\rho(\underline{\xi}),
\end{equation}
where we reparametrise the argument in phase space coordinates, 
$\alpha_j = \frac{1}{\sqrt{2}}(x_j+ip_j)$, 
and $\underline{\alpha}\cdot\underline{\xi}^\dagger = \sum_j \alpha_j\overline{\xi}_j$
is the Hermitian inner product of the complex coordinate tuples. 
This is a real-valued function, and its normalisation is chosen 
in such a way that 
\begin{equation}
  \int {\rm d}^m x{\rm d}^m p\, W_\rho(x,p) = \tr\rho,
\end{equation}
hence for a state we can address it as a quasi-probability function
as it integrates to $1$, 
and if the Wigner function is positive it is a genuine probability density.
In general, can be expressed as an expectation value, cf. \cite{Barnett,cahill}, 
\begin{equation}
  \label{eq:Wigner-expectation}
  W_\rho(x,p)
    = \pi^{-m} \tr \rho D(\underline{\alpha})(-1)^{N_1+\ldots+N_m}D(\underline{\alpha})^\dagger.
\end{equation}
where as before $\alpha_j = \frac{1}{\sqrt{2}}(x_j+ip_j)$.
It shows that $W_\rho$ is well-defined and indeed a continuous 
bounded function for all trace class operators: indeed, 
$|W_\rho(x,p)| \leq \pi^{-m} \|\rho\|_1$.
The above formula can be used to give meaning to more general operators (such 
as POVM elements); for instance the Wigner function of the identity operator 
is a constant, $W_\1 = (2\pi)^{-m}$.
The Wigner transformation preserves the Frobenius (Hilbert-Schmidt) inner
product,
\begin{equation}
  \label{eq:Frobenius}
  \tr\rho\sigma = (2\pi)^m \int {\rm d}^m x {\rm d}^m p\, W_\rho(x,p)W_\sigma(x,p).
\end{equation}

Unitary transformations of the Hilbert space preserve the canonical 
commutation relations; but the subset of unitaries that map the Lie algebra 
of the canonical variables, which is $\operatorname{span}\{\1,x_j,p_k\}$,
to itself, are called \emph{Gaussian} unitaries. We address them also 
as ``linear'' transformations, since they are correspond to an affine linear 
map of phase space, and are described by a displacement vector and a symplectic matrix.
Gaussian channels are precisely the completely positive and trace preserving 
(cptp) maps taking Gaussian states to Gaussian states. It is a fundamental fact 
that a quantum channel $\cN:\cS(\cH)\rightarrow\cS(\cH')$ is Gaussian if and 
only if it has a Gaussian unitary dilation $U$, with the environment initialised 
in the vacuum state:
\begin{equation}
  \cN(\rho) = \tr_E U\left(\rho\ox\proj{0}^{\ox\ell}\right)U^\dagger,
\end{equation}
where the environment has $\ell$ modes.

For the following, we need the (Glauber-Sudarshan) coherent states, also 
known as minimal dispersion states $\ket{\alpha}$, which are eigenstates 
of the annihilation operator: $a\ket{\alpha} = \alpha\,\ket{\alpha}$, for 
$\alpha\in\CC$. This defines the states uniquely, and one can show that
they are related by displacements: $\ket{\alpha} = D(\alpha)\ket{0}$,
where $\ket{0}$ is both the coherent state corresponding to $\alpha=0$
and the vacuum, i.e. the ground state of the Hamiltonian, in other 
words the zeroth Fock state. In the Fock basis,
\begin{equation}
  \label{eq:coherentstate}
  \ket{\alpha} = e^{-\frac12 |\alpha|^2} \sum_{n=0}^\infty \frac{\alpha^n}{\sqrt{n!}}\ket{n},
\end{equation}
a relation that reassuringly shows that the coherent states are 
well-defined unit vectors, written in a genuine orthonormal basis. 
However, what is more relevant are the
following expressions for the first and second moments. 
For $\alpha = \alpha_R + i\alpha_I$ written in terms of real and imaginary
parts, 
\begin{align}
  \bra{\alpha} x \ket{\alpha} &= \alpha_R\sqrt{2},\ \bra{\alpha} p \ket{\alpha} = \alpha_I\sqrt{2},\\
  \bra{0} x^2 \ket{0}         &= \bra{0} p^2 \ket{0} = \frac12,
\end{align}
the latter ``vacuum fluctuations'' consistent with the Heisenberg-Robertson 
uncertainty relation. Furthermore, the inner product, easily confirmed from 
the expansion in the Fock basis,
\begin{equation}
  |\bra{\alpha}\beta\rangle|^2 = e^{-|\alpha-\beta|^2}.
\end{equation}
And finally, we record
\begin{equation}
  \label{eq:hetero}
  \frac{1}{\pi} \int {\rm d}\alpha\, \proj{\alpha} = \1,
\end{equation}
showing that the family of operators $\frac{{\rm d}\alpha }{\pi}\proj{\alpha}$
forms a POVM, known as heterodyne measurement.

\section{State discrimination by \protect\\ Gaussian measurements}
\label{sec:GOCC}
Now that we have the Bosonic formalism in place, we can discuss the 
problem of binary hypothesis testing under Gaussian restrictions on the
measurement. 
Indeed, going back to Eqs.~(\ref{eq:HelstromHolevo}) and (\ref{eq:tracenorm})
in the introduction, almost any restriction $\mathbb{M}$
on the set of possible measurements, be they physically motivated or 
purely mathematical, results in a larger error probability than the Helstrom 
expression, which is most conveniently expressed in terms of a distinguishability 
norm on states: 
\begin{equation}
  \label{eq:M-norm}
  \min_{(M,\1-M)\in\mathbb{M}} P_e 
                =: \frac12 \left( 1-\frac12\|\rho_0-\rho_1\|_{\mathbb{M}} \right).
\end{equation}
How to define the set $\mathbb{M}$ appropriately and what exactly is 
necessary for it to define a norm is explained in detail in \cite{MWW}. 
An example exceedingly well-studied in quantum information theory is 
the set of measurements implemented by local operations and classical 
communication (LOCC) in a bi- or multi-partite system, as well as its 
relaxations separable POVM elements (SEP) and POVM elements with positive 
partial transpose (PPT), cf.~\cite{MWW,LOCC:always}

Here, we consider restrictions motivated from the fact that Gaussian 
operations are distinguished among the ones allowed by quantum mechanics 
generally, following Takeoka and Sasaki \cite{TakeokaSasaki}. Concretely,
we are interested in the measurements 
implemented by any sequence of partial Gaussian POVMs and classical
feed-forward (Gaussian operations and classical computation, GOCC).
Very much like LOCC, there is no concise way of writing down a general 
GOCC transformation, but for a binary measurement the prescription is
as follows.

\begin{definition}
\label{defi:GOCC}
A \emph{GOCC measurement protocol} on $m$ modes consists of the repetition of the 
following steps, for $r=1,\ldots,R$ (``rounds''), 
after initially setting $\xi_0=\emptyset$ and $m_{\emptyset}=m$. 
Here, $\xi_{r-1}$ is the collection of all measurement outcomes prior to
round $r$.
\begin{itemize}
  \item[{(r.1)}] create a number $k_{\xi_{r-1}}$ of Bosonic modes in the vacuum state;

  \item[{(r.2)}] perform a Gaussian unitary $U_{\xi_{r-1}}$ on the 
     $m_{\xi_{r-1}}+k_{\xi_{r-1}}$ modes; 

  \item[{(r.3)}] perform homodyne detection on the last $\ell_{\xi_{r-1}}$ modes, 
     keeping the first $m_{\xi_{r}} := m_{\xi_{r-1}}+k_{\xi_{r-1}}-\ell_{\xi_{r-1}}$;
     call the outcome $\underline{x}^{(r)} = x_1^{(r)}\ldots x_{\ell_{\xi_{r-1}}}^{(r)}$
     and set $\xi_{r} := \{\xi_{r-1},\underline{x}^{(r)}\}$. 
\end{itemize}

Each $r$ is called a ``round'', and in the $R$-th round all remaining 
modes are measured, i.e. $\ell_{\xi_{r-1}} = m_{\xi_{r-1}}+k_{\xi_{r-1}}$. 
The final measurement outcome is a measurable function $f(\xi_R) \in \Omega$, 
taking values in a prescribed set $\Omega$, 
which for simplicity we assume to be discrete.

This defines a POVM $(M_\omega:\omega\in\Omega)$, and every POVM
that arises in the above way, or as a limit of such POVMs in the 
strong 
topology is called a \emph{GOCC POVM}.
\end{definition}

\medskip
Now, returning to equiprobable hypotheses $\rho_0$ and $\rho_1$,
and enforcing the POVMs to be implemented by GOCC protocols,
we arrive at the GOCC norm:
\begin{equation}
  \label{eq:GOCC-norm}
  \inf_{(M,\1-M)\atop \text{GOCC POVM}} P_e 
                =: \frac12 \left( 1-\frac12\|\rho_0-\rho_1\|_\GOCC \right).
\end{equation}
Note that $\|\cdot\|_\GOCC$ is indeed a norm, since the set of 
GOCC measurements is tomographically complete. Indeed, heterodyne detection 
on every available mode is a tomographically complete measurement, 
meaning that for every pair of distinct quantum states, there exists
a binary coarse graining of the heterodyne detection outcomes 
that discriminates the states with some non-zero bias. 

\begin{remark}
In the definition of a GOCC protocol, we could have allowed the 
$k_{\xi_{r-1}}$ ancillary modes to be prepared in any Gaussian state
in step (r.1), 
but that does not add any more generality, since every Gaussian state can 
be prepared from the vacuum by a suitable Gaussian unitary.
We could also have allowed an arbitrary Gaussian quantum channel in 
step (r.2), but again that does not add any more generality since 
every Gaussian channel has a dilation to a Gaussian unitary with 
an environment prepared in the vacuum state. 
Finally, in step (r.3) we could have allowed any Gaussian measurement, 
but every Gaussian measurement can be implemented by adjoining suitable 
ancilla modes in the vacuum, performing a Gaussian unitary and a homodyne 
measurement. 

From the point of view of the discussion of classes of operations, 
of which measurements are a special case, it is interesting to distinguish 
certain subclasses of GOCC: what we actually have defined are the 
measurements implemented by a GOCC protocol with finitely many rounds,
as well as the closure of this set. One could also define the POVMs 
implemented by a GOCC protocol with unbounded rounds (but probability 
1 to stop), which would sit between the former two, 
cf.~\cite{LOCC:always} for the case of LOCC. While it is interesting 
to study these three classes, in particular whether they coincide or 
are separated (as they are in the analogous case of LOCC \cite{LOCC:always})
this is beyond the scope of the present work. Indeed, for the case of 
hypothesis testing, thanks to the infimum in the error probability, 
all three classes will give rise to the same GOCC norm. 
\end{remark}

\medskip
An elementary observation about GOCC is that the fine-grained measurement 
(i.e. before coarse-graining to a discrete POVM) consists of operators
each of which is a positive scalar multiple of a Gaussian pure state. 
In particular, they have non-negative Wigner function, and because the
coarse-graining amounts to summing POVM elements, also the final 
POVM has non-negative Wigner functions. 
We thus call a binary POVM $(M,\1-M)$ with non-negative Wigner functions 
$W_{M}$ and $W_{\1-M}$ a \emph{W+ POVM}, and denote their set $\mathbb{W}_+$. 

Just as the restriction to GOCC leads to the distinguishability 
norm $\|\cdot\|_\GOCC$ [Eq.~(\ref{eq:GOCC-norm})], the restriction to 
W+ POVMs gives rise to the distinguishability norm $\|\cdot\|_\Wplus$:
\begin{equation}
  \label{eq:Wplus-norm}
  \inf_{(M,\1-M)\atop \text{W+ POVM}} P_e 
                =: \frac12 \left( 1-\frac12\|\rho_0-\rho_1\|_\Wplus \right).
\end{equation}

Since every GOCC measurement is automatically W+, we have by definition
\begin{equation}
  \label{eq:tower}
  \|\rho_0-\rho_1\|_\GOCC \leq \|\rho_0-\rho_1\|_\Wplus \leq \|\rho_0-\rho_1\|_1. 
\end{equation}
The rest of the paper is concerned with the comparison of these norms.
The questions guiding us are: are they different, and how large are the gaps?

\section{Separation between GOCC and unrestricted measurements}
\label{sec:GOCC-vs-ALL}
Our first result shows a simple upper bound on the GOCC and W+ 
distinguishability norms in terms of the Wigner functions of the 
two states. 

\begin{lemma}
  \label{lemma:W+vs-L1}
  For any two states $\rho_0$ and $\rho_1$ of an $m$-mode system,
  with associated Wigner functions $W_0$ and $W_1$, respectively,
  \[\begin{split}
    \| \rho_0-\rho_1 \|_\GOCC &\leq \| \rho_0-\rho_1 \|_\Wplus           \\
                              &\leq \| W_0-W_1 \|_{L^1}                  \\
                              &=    \int {\rm d}^m x{\rm d}^m p\, |W_0(x,p)-W_1(x,p)|.
  \end{split}\]
\end{lemma}
Note that, unlike the inequalities (\ref{eq:tower}), the third term 
in the chain is not a trace norm of density matrices, but an $L^1$ norm 
of real functions, which we may interpret as generalised densities.

\begin{proof}
Only the second inequality remains to be proved. Consider 
any W+ POVM $(M,\1-M)$, meaning that the stochastic response functions 
$F = (2\pi)^m W_M$ and $1-F = (2\pi)^m W_{\1-M}$
are bounded between $0$ and $1$. By the Frobenius inner product 
formula for the Wigner function, Eq.~(\ref{eq:Frobenius}), we have
\begin{equation}
  \label{eq:W-expectation}
  \tr(\rho_0-\rho_1)M = \int {\rm d}^m x{\rm d}^m p\, \bigl(W_0(x,p)-W_1(x,p)\bigr)F(x,p),
\end{equation}
where the left hand side appears in Eq.~(\ref{eq:HelstromHolevo}), its 
supremum over W+ POVMs being $\frac12\|\rho_0-\rho_1\|_\Wplus$; while 
the right hand side is upper bounded by $\frac12 \| W_0-W_1 \|_{L^1}$, 
where we made use of the fact that
$\int {\rm d}^m x{\rm d}^m p\, (W_0(x,p)-W_1(x,p)) = \tr(\rho_0-\rho_1) = 0$.
\end{proof}

\begin{remark}
The lemma assumes measurements with W+ POVMs, but it gives interesting 
information also in cases where the POVM has some limited Wigner negativity.
Namely, looking at Eq.~(\ref{eq:W-expectation}), we subsequently use that 
$0 \leq F(x,p) \leq 1$, which is the property W`+ of the measurement.

If we do not have ``too much'' Wigner negativity in the measurement operators, 
this could be expressed by a bound $|2F(x,p)-1| \leq B$, and then we 
would get 
\begin{equation}
  \bigl| \tr(\rho_0-\rho_1)M \bigr| \leq \frac{B}{2} \| W_0-W_1 \|_{L^1}.
\end{equation}
The right hand side can still be small when the $L^1$-distance is really
small, and at the same time $B$ not too large. In the next section we 
shall see an example of this.
\end{remark}

\medskip
As one might expect, the inequality in Lemma \ref{lemma:W+vs-L1} is often crude, 
or even trivial since one can find states where the right had side exceeds $2$.
However, if $\rho_0$ and $\rho_1$ are both states with non-negative
Wigner function, for instance probabilistic mixtures of Gaussian states, 
then $W_0$ and $W_1$ are \emph{bona fide} probability densities, 
and the right hand side is $\leq 2$. 
In that case, we have the following corollary for the quantum Chernoff 
coefficient when measurements are restricted to GOCC or W+ POVMs.
Recall that the Chernoff coefficient is the exponential rate of the 
minimum error probability in distinguishing two i.i.d. hypotheses. 
I.e., in the case of two quantum states
\begin{equation}
  \xi(\rho_0,\rho_1) 
     := \lim_{n\rightarrow\infty} 
           - \frac1n \ln \left(1-\frac12\left\|\rho_0^{\ox n}-\rho_1^{\ox n}\right\|_1\right),
\end{equation}
which generalises the analogous question for probability distributions
\cite{Chernoff}. 
Amazingly, there is a formula for this exponent \cite{q-Chernoff}, 
generalising in its turn the classical answer \cite{Chernoff}:
\begin{equation}
  \xi(\rho_0,\rho_1) = -\ln \inf_{0<s<1} \tr\rho_0^s\rho_1^{1-s}.
\end{equation}

Just as the distinguishability norm under a restriction, we can then 
define the constrained Chernoff coefficient, if the restriction 
$\mathbb{M}$ describes a subset of POVMs for each number of elementary systems:
\begin{equation}
  \xi_{\mathbb{M}}(\rho_0,\rho_1) 
     := \lim_{n\rightarrow\infty} 
           - \frac1n \ln \left(1-\frac12\left\|\rho_0^{\ox n}-\rho_1^{\ox n}\right\|_{\mathbb{M}}\right).
\end{equation}

\begin{corollary}
  \label{cor:Chernoff}
  For two states $\rho_0$ and $\rho_1$ with non-negative Wigner functions 
  $W_0,\, W_1 \geq 0$ (meaning that they are probability density functions),
  \[
    \xi_\GOCC(\rho_0,\rho_1) \leq \xi_\Wplus(\rho_0,\rho_1) 
                             \leq \xi(W_0,W_1),
  \]
  where according to Chernoff's theorem \cite{Chernoff}, 
  \[
    \xi(W_0,W_1) = -\ln \inf_{0<s<1} \int {\rm d}^m x{\rm d}^m p\, W_0(x,p)^s W_1(x,p)^{1-s}
  \]
  is the classical Chernoff coefficient of the probability distributions 
  $W_0$ and $W_1$. 
  \qed
\end{corollary}
As in Lemma \ref{lemma:W+vs-L1}, the third term in the chain is not a quantity 
of density matrices, but of classical probability densities. 


\medskip
It is not difficult to find examples where the bounds of the Lemma and its
Corollary are exactly tight, among them the case of two coherent states
studied originally by Takeoka and Sasaki \cite{TakeokaSasaki}.

\begin{example}
  \label{example:takeoka-sasaki}
  Consider $\rho_0$ and $\rho_1$ as two coherent states of a single mode, 
  say $\rho_0=\proj{+\alpha}$, $\rho_1=\proj{-\alpha}$ for $\alpha > 0$. 
  Then, 
  \begin{equation}
    \frac12 \|\rho_0-\rho_1\|_1 = \sqrt{1-e^{-4\alpha^2}}.
  \end{equation}
  while by Lemma \ref{lemma:W+vs-L1},
  \begin{equation}
    \frac12 \|\rho_0-\rho_1\|_\GOCC = \frac12 \|\rho_0-\rho_1\|_\Wplus
                                    = \operatorname{erf}(\alpha\sqrt{2}),
  \end{equation}
  with the error function 
  $\operatorname{erf}(x) = \frac{2}{\sqrt\pi} \int_0^x {\rm d}x\,e^{-x^2}$.
  The equality follows from homodyning the $x$-coordinate and
  deciding depending on the sign of the measurement outcome. 
  The norms are compared in Fig.~\ref{fig:GOCC-vs-tr}.

  Furthermore, in the asymptotic i.i.d. setting of the Chernoff bound, 
  \begin{equation}
    \xi(\rho_0,\rho_1) = -\ln F(\rho_0,\rho_1)^2 = 4 |\alpha|^2,
  \end{equation}
  while by Corollary \ref{cor:Chernoff},
  \begin{equation}
    \xi_\GOCC(\rho_0,\rho_1) = \xi_\Wplus(\rho_0,\rho_1) = 2 |\alpha|^2.
  \end{equation}
  The equality follows from homodyning each mode separately in the $x$ 
  direction, and classical post-processing.
\end{example}

\begin{figure}[ht]
  \includegraphics[width=7.5cm]{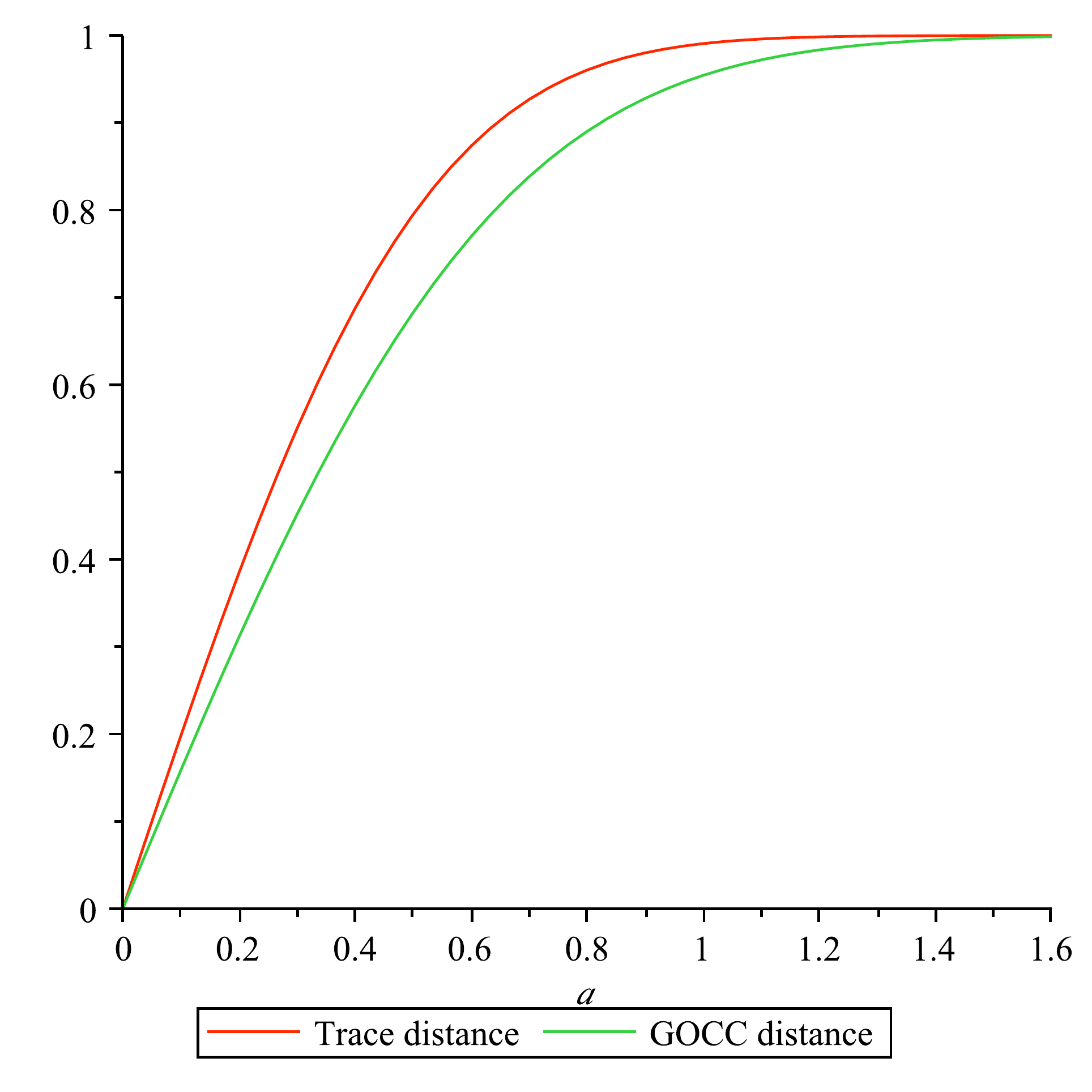}
  \caption{Plot of the trace distance (red) versus the GOCC distance (green)
           against $\alpha$ on the horizontal axis. While there is a
           nonzero gap for all $\alpha > 0$, it vanishes for asymptotically
           large and small displacements, as expected. 
           The largest difference between $\frac12 \|\rho_0-\rho_1\|_1$ and 
           $\frac12 \|\rho_0-\rho_1\|_\GOCC$ is $\approx 0.11$, occurring
           at $\alpha \approx 0.45$.}
  \label{fig:GOCC-vs-tr}
\end{figure}

\begin{example}  
  More generally, for any one-mode Gaussian state and its displacement 
  along one of the principal axes of the covariance matrix,
  \begin{equation}
    \|\rho_0-\rho_1\|_\GOCC = \|\rho_0-\rho_1\|_\Wplus
                            = \|W_0-W_1\|_{L^1},
  \end{equation}
  and the latter can be expressed in terms of the error function
  and the shared variance of the two states in the direction of the 
  displacement connecting them. 

  The equality follows from homodyning in the direction of the 
  line connecting the two first moment vectors in phase space, and deciding 
  depending on which of the two points is closer to the outcome. 
\end{example}

\section{Data hiding secure against Gaussian attacker}
\label{sec:hiding}
As soon as we realize that it is possible to get large gaps between 
$\|\cdot\|_1$ and $\|\cdot\|_\GOCC$, we have to ask ourselves, just how 
large the gap can be.
In particular, is it possible to find state pairs which are
almost maximally distant in the trace norm, yet almost indistinguishable 
in the GOCC norm? 
In other words, can we protect the information against an adversary who 
attempts the hypothesis testing on the two states but with only access
to Gaussian operations and classical communication? This is the definition 
of data hiding, first explored 
in the context of the LOCC restriction, 
and then later abstractly for an arbitrary restriction on the possible 
measurements. 

Next we shall show that data hiding is possible also under GOCC, 
at least when going to multiple modes. 

\begin{theorem}
  \label{thm:GOCC-hiding}
  Let $\overline{E} > 0$. Then, there is a constant $c>0$ such 
  that for all sufficiently large integers $m$ there exist $m$-mode states 
  $\rho_0$ and $\rho_1$, each a mixture of a finite set of coherent states 
  and with average energy (photon number) per mode bounded by $\overline{E}+o(1)$, 
  such that
  \begin{align}
    \label{eq:random-tracenorm}
    \frac12 \| \rho_0\!-\!\rho_1 \|_1 &\geq 1-e^{-cm}, \\
    \label{eq:random-Wplusnorm}
    e^{-cm} &\geq \frac12 \| \rho_0\!-\!\rho_1 \|_\Wplus 
             \geq \frac12 \| \rho_0\!-\!\rho_1 \|_\GOCC.
  \end{align}
\end{theorem}
\begin{proof}
Consider the $m$-mode coherent states 
$\ket{\underline{\alpha}^{(\lambda)}} 
  = \ket{\alpha^{(\lambda)}_1}\ket{\alpha^{(\lambda)}_2}\cdots\ket{\alpha^{(\lambda)}_m}$ 
($\lambda=1,\ldots,2L$), where the parameters $\alpha^{(\lambda)}_j\in\CC$ are 
chosen i.i.d according to a normal distribution with mean $0$ and variance 
$\EE |\alpha^{(\lambda)}_j|^2 = \overline{E}$. 
Then define 
\begin{equation}\begin{split}
  \rho_0 &= \frac{1}{L} \sum_{\lambda=1}^L \proj{\underline{\alpha}^{(2\lambda)}}, \\
  \rho_1 &= \frac{1}{L} \sum_{\lambda=1}^L \proj{\underline{\alpha}^{(2\lambda-1)}},
\end{split}\end{equation}
so these are random states. Note that with high probability, 
indeed asymptotically converging to $1$ as $m\gg 1$, 
both have their photon number per mode bounded by $\overline{E}+o(1)$.
Also, 
\begin{equation}
  \EE \rho_0 = \EE \rho_1 = \gamma(\overline{E})^{\ox m}, 
\end{equation}
where $\gamma(\overline{E}) = (1-e^{-\beta}) e^{-\beta N}$ is the thermal state 
of a single Bosonic mode of mean photon number $\overline{E}$, 
i.e. with $\beta = \ln\left(1+\frac{1}{\overline{E}}\right)$.

The rest of the proof will consist in showing that we can fix $L$ in such a way 
that with probability close to $1$, $\rho_0$ and $\rho_1$ are distinguishable 
except with exponentially small error probability, and that with probability 
close to $1$, the Wigner functions $W_0$ and $W_1$ are exponentially close to 
$W_{\gamma(\overline{E})}^{\ox m}$, the Wigner function of 
$\gamma(\overline{E})^{\ox m}$, in the total variational distance. 

\medskip
\emph{Eq.~(\ref{eq:random-tracenorm}):} The ensemble of coherent states 
$\ket{\underline{\alpha}^{(\lambda)}}$ is the well-studied random 
coherent state modulation of the noiseless Bosonic channel with input 
power (photon number) $\overline{E}$, 
whose classical capacity is well-known \cite{pure-loss-C},
with the strong converse proved in \cite{WiWi:pure-loss}.
\begin{equation}\begin{split}
  C(\id,\overline{E}) &= g(\overline{E}) \\
                      &= (\overline{E}+1)\ln(\overline{E}+1) - \overline{E}\ln\overline{E} \\
                      &= \ln(1+\overline{E}) + \overline{E}\ln\left(1+\frac{1}{\overline{E}}\right).
\end{split}\end{equation}
Thus, when $2L\leq e^{m\bigl(C(\id,\overline{E})-\delta\bigr)}$, it follows 
from the Holevo-Schumacher-Westmoreland theorem \cite{Holevo:C,SchumacherWestmoreland:C,Holevo:C-E}
that with probability close to $1$, there exists a POVM $(D_\lambda)_{\lambda=1}^{2L}$ 
that decodes $\lambda$ reliably from the state $\proj{\underline{\alpha}^{(\lambda)}}$: 
\begin{equation}
  \frac{1}{2L}\sum_{\lambda=1}^{2L} \tr \proj{\underline{\alpha}^{(\lambda)}} D_\lambda \geq 1-e^{-c'm},
\end{equation}
with a suitable constant $c'>0$ and for all sufficiently large $m$. 
Thus, with $\rho_i$ ($i=0,1$) as defined above and
\begin{equation}
  M_i = \sum_{\lambda=1}^L D_{2\lambda-i} \quad (i=0,1),
\end{equation}
it follows
\begin{equation}
  \frac12 \tr \rho_0 M_0 + \frac12 \tr \rho_1 M_1 \geq 1-e^{-c'm},
\end{equation}
which implies Eq.~(\ref{eq:random-tracenorm}).

\medskip
\emph{Eq.~(\ref{eq:random-Wplusnorm}):} The Wigner functions 
$W_{\underline{\alpha}^{(\lambda)}}$ of the coherent states 
$\proj{\underline{\alpha}^{(\lambda)}}$ are $2m$-dimensional real 
Gaussian probability densities centered at $\underline{z}^{(\lambda)}$, 
where $z^{(\lambda)}_{2j-1} = \Re \alpha^{(\lambda)}_j\sqrt{2}$ and
$z^{(\lambda)}_{2j} = \Im \alpha^{(\lambda)}_j\sqrt{2}$ are the 
rescaled real and imaginary part of $\underline{\alpha}^{(\lambda)}_j$, 
respectively;
they have variance $\frac12$ in each direction. 
We read them as output distributions of an i.i.d.~additive white Gaussian noise 
(AWGN) channel on $2m$ inputs $\underline{z}^{(\lambda)}$,
and with noise power $\frac12$. Note that all $z^{(\lambda)}_j$ are themselves 
Gaussian distributed random variables with $\EE z^{(\lambda)}_j = 0$ 
and $\EE |z^{(\lambda)}_j|^2 = \overline{E}$.
This channel, which we denote $\widetilde{W}$ since its output distributions 
come from the Wigner functions of the coherent states 
$\underline{\alpha}^{(\lambda)}$, thanks to Shannon's famous formula with 
the signal-to-noise ratio has the capacity 
\begin{equation}
  C(\widetilde{W},\overline{E}) = \frac12 \ln(1+2\overline{E}).
\end{equation}
Thus, by the theory of approximation of output statistics \cite{HanVerdu:AOS},
adapted to the AWGN channel \cite{HanVerdu-AWGN}, it follows that when 
$2L\geq e^{2m\bigl(C(\widetilde{W},\overline{E})+\delta\bigr)}$, then with probability 
close to $1$
\begin{equation}
  \left\| W_i - W_{\gamma(\overline{E})}^{\ox m} \right\|_{L^1} \leq \frac12 e^{-c''m},
\end{equation}
for $i=0,1$, with a suitable constant $c''>0$ and for all sufficiently large $m$.
See~\cite[Thm.~6.7.3]{Han:InfoSpec} for the concrete statement. 
Hence, by the triangle inequality and Lemma \ref{lemma:W+vs-L1},
we get Eq.~(\ref{eq:random-Wplusnorm}). 

\medskip
It remains to put the two parts together: We observe that 
$2C(\widetilde{W},\overline{E}) < C(\id,\overline{E})$ for all $\overline{E} > 0$.
Indeed, a well-known elementary inequality states
\begin{equation}
  \ln(1+t) \geq \frac{t}{1+t},
\end{equation}
which we apply to $t=\frac{1}{\overline{E}}$, yielding
\begin{equation}
  E\ln\left(1+\frac{1}{\overline{E}}\right) 
     \geq \overline{E}\frac{\frac{1}{\overline{E}}}{1+\frac{1}{\overline{E}}}
     =    \frac{\overline{E}}{1+\overline{E}}
     >    \ln\left(1+\frac{\overline{E}}{1+\overline{E}}\right), 
\end{equation}
which is equivalent to the claim. 
This means that we can choose $\delta>0$ such that 
\(
  2C(W,\overline{E})+2\delta < C(\id,\overline{E})-\delta, 
\)
meaning we can satisfy 
\begin{equation}
  e^{2m\bigl(C(\widetilde{W},\overline{E})+\delta\bigr)} \leq 2L \leq e^{2m\bigl(C(\widetilde{W},\overline{E})+\delta\bigr)}
\end{equation}
simultaneously for all sufficiently large $m$. Finally, setting $c=\min\{c',c''\}$ 
concludes the proof. 
\end{proof}

\begin{remark}
While we didn't make any attempt to give a numerical value for $c$ 
(which is a function of $\overline{E}$), in principle it can be 
extracted from the HSW coding theorem for the noiseless Bosonic 
channel and the resolvability coding theorem for the AWGN channel.

Likewise, we presented the theorem as an asymptotic result, but the 
proofs of the two coding theorems will yield finite values of $m$ for 
which the constructions work with probability $>\frac34$, 
and so we get the existence of the data hiding states for that number of 
modes. 
\end{remark}

\begin{corollary}
  \label{cor:big-Chernoff}
  For the two $m$-mode states $\rho_0$ and $\rho_1$ from Theorem \ref{thm:GOCC-hiding},
  \[
    \xi(\rho_0,\rho_1) \geq \frac{c}{2}m - \ln\sqrt{2},
  \]
  whereas 
  \[\begin{split}
    \xi_\GOCC(\rho_0,\rho_1) &\leq \xi_\Wplus(\rho_0,\rho_1) \\
                             &\leq \xi(W_0,W_1)              \\
                             &\leq -2\ln\left(1-e^{-cm}\right) \sim 2 e^{-cm}.
  \end{split}\]
\end{corollary}
\begin{proof}
With $\epsilon = e^{-cm}$ as in Theorem \ref{thm:GOCC-hiding}, 
we use the Fuchs-van de Graaf relation between trace distance 
and fidelity \cite{FvdG}:
\begin{equation}
  1-F(\rho_0,\rho_1) \leq \frac12 \|\rho_0-\rho_1\|_1 \leq \sqrt{1-F(\rho_0,\rho_1)^2},
\end{equation}
where the mixed-state fidelity is given 
\begin{equation}
  F(\rho_0,\rho_1) := \|\sqrt{\rho_0}\sqrt{\rho_1}\|_1. 
\end{equation}

Now, we get first $1-F(\rho_0,\rho_1)^2 \geq (1-\epsilon)^2$. And then
we can estimate: 
\begin{equation}\begin{split}
  e^{-\xi(\rho_0,\rho_1)} &=    \inf_{0<s<1} \tr\rho_0^s\rho_1^{1-s} \\
                          &\leq \tr\sqrt{\rho_0}\sqrt{\rho_1}       \\
                          &\leq \|\sqrt{\rho_0}\sqrt{\rho_1}\|_1   \\
                          &=    F(\rho_0,\rho_1) 
                           \leq \sqrt{2\epsilon}.
\end{split}\end{equation}

Secondly, we get $1-\epsilon \leq F(W_0,W_1) = F(\omega_0,\omega_1)$,
with suitable purifications $\omega_i$ of $W_i$, according to Uhlmann's theorem. 
By Corollary \ref{cor:Chernoff}, 
\begin{equation}\begin{split}
    \xi_\GOCC(\rho_0,\rho_1) &\leq \xi_\Wplus(\rho_0,\rho_1) \\
                             &\leq \xi(W_0,W_1)               \\
                             &\leq \xi(\omega_0,\omega_1)      \\
                             &=    -\ln F(\omega_0,\omega_1)^2 \\
                             &=    -\ln F(W_0,W_1)^2           
                              \leq -2\ln(1-\epsilon),
\end{split}\end{equation}
where in the third line we have used the monotonicty of the Chernoff 
coefficient under partial traces, and in the fourth line the formula
for the Chernoff coefficient for pure states. 
\end{proof}

\section{Lower bounds on distinguishability under \protect\\ 
         GOCC and W+ measurements}
\label{sec:lower-bounds}
So far, we have seen examples of separations, including large ones, 
between the trace norm and GOCC and W+ norms. 
Especially about the construction in the previous section
we can ask, whether and in which sense it uses the available resources
optimally: these would be the number of modes and the energy. 
Here we show lower bounds on the distinguishability of general
states when restricted to W+, compared to the trace norm. 
They are motivated by similar studies under the LOCC, SEP or PPT 
constraint, or an abstract constraint on the allowed measurements
\cite{MWW,LW}, see also \cite{ultimate}.

\begin{proposition}
  \label{prop:W+lowerbound}
  For any two $m$-mode states $\rho_0$ and $\rho_1$,
  \[
    \| \rho_0-\rho_1 \|_\Wplus \geq 2^{-m-1}\| \rho_0-\rho_1 \|_2^2 = 2^{-m-1}\tr(\rho_0-\rho_1)^2.
  \]
\end{proposition}
\begin{proof}
We will write down a specific W+ POVM that achieves the r.h.s. as its 
statistical distance. In fact, with $\Delta = \rho_0-\rho_1$, for 
our POVM $(M,\1-M)$ we make the ansatz  
\begin{equation}\begin{split}
  M    &= \frac12(\1 + \eta\Delta), \\
  \1-M &= \frac12(\1 - \eta\Delta),
\end{split}\end{equation}
with a suitable constant $\eta>0$ to ensure that not only is this a POVM 
(for which it is enough that $\eta\leq 1$), but a W+ POVM. For that purpose, 
recall that $W_\1 = (2\pi)^{-m}$.
Recall furthermore that the Wigner functions of states are 
bounded, $|W_\rho(x,p)| \leq \pi^{-m}$, 
see Eq.~(\ref{eq:Wigner-expectation}).

This means that $|W_\Delta(x,p)| \leq 2 \pi^{-m}$, 
and so $W_{M/\1-M} = \frac12 W_\1 \pm \frac12 \eta W_\Delta \geq 0$
is guaranteed by letting $\eta = 2^{-m-1}$.

Thus we have $2M-\1=\eta\Delta$, and can calculate
\begin{equation}
  \|\rho_0-\rho_1\|_\Wplus \geq \tr \Delta(2M-\1) = 2^{-m-1}\|\rho_0-\rho_1\|_2^2,
\end{equation}
concluding the proof.
\end{proof}

\begin{corollary}
  \label{cor:W+lowerbound}
  Consider two $m$-mode states $\rho_0$ and $\rho_1$, with average energy
  (photon number) per mode bounded by $\overline{E}$ and 
  $\|\rho_0-\rho_1\|_1 \geq t > 0$. Then, with $t=4c+r$,
  \[
    \| \rho_0-\rho_1 \|_\Wplus \geq r^2 2^{-m-1}\left(1+\frac{\overline{E}}{c^2}\right)^{-m} \!\!\!\!.
  \]

  Thus, to achieve the kind of separation as in Theorem~\ref{thm:GOCC-hiding},
  between a ``large'' trace norm and ``small'' W+ norm, with bounded energy per mode, 
  their number necessarily has to grow; or else, the energy per mode has to
  grow very strongly.
\end{corollary}
\begin{proof}
Construct the projector $P$ onto the space of all $m$-mode number 
states with photon number $\leq m \frac{\overline{E}}{c^2}$. 
Denote $D = \operatorname{rank} P$.
By the assumption of the energy bound, and Markov's inequality, 
\begin{equation}
  \tr\rho_0 P,\ \tr\rho_1 P \geq 1-c^2. 
\end{equation}
Hence, by the gentle measurement lemma \cite{Winter:qstrong},
\begin{equation}
  \|\rho_0 - P\rho_0 P\|_1,\ \|\rho_1 - P\rho_1 P\|_1 \leq 2c,
\end{equation}
and so by the triangle inequality $\|P\rho_0 P-P\rho_1 P\|_1 \geq t-4c = r$.

Now, 
\begin{equation}\begin{split}
  \|\rho_0-\rho_1\|_2 &\geq \|P\rho_0 P-P\rho_1 P\|_2 \\
                      &\geq \frac{1}{\sqrt{D}}\|P\rho_0 P-P\rho_1 P\|_1 
                       \geq \frac{r}{\sqrt{D}},
\end{split}\end{equation}
where the first inequality follows from the fact that the Frobenius norm 
squared is the sum of the modulus-squared of the all the matrix entries, 
and the projector $P$ simply gets rid of some of those; the second is the 
well-known comparison between (Schatten) $1$- and $2$-norms on a $D$-dimensional 
space.  
Thus, by Proposition \ref{prop:W+lowerbound} we have 
\begin{equation}
  \|\rho_0-\rho_1\|_\Wplus \geq 2^{-m-1}\|\rho_0-\rho_1\|_2^2 \geq r^2 2^{-m-1}\frac{1}{D},
\end{equation}
and it remains to control $D$. Note that by its definition, it has an 
exact expression as a binomial coefficient,
\begin{equation}
  D =    {\left\lfloor \frac{\overline{E}}{c^2} \right\rfloor + m \choose m}
    \leq {\frac{\overline{E}}{c^2} + m \choose m}
    \leq \left(1+\frac{\overline{E}}{c^2}\right)^{m},
\end{equation}
concluding the proof.
\end{proof}

\medskip
We believe that a lower bound like the one of Corollary \ref{cor:W+lowerbound} 
should hold for the GOCC norm, too. To get such a bound, we need to find a 
``pretty good'' Gaussian measurement to distinguish two given states. 

A possible strategy might be provided by \cite[Thms.~13 and 14]{MWW}, 
where it is shown that in dimension $D<\infty$, a fixed rank-one POVM $\cM$
whose elements form a (weighted) $2$-design, provides a bound
\begin{equation}
  \|\cdot\|_{\cM} \geq \frac{1}{2D+2}\|\cdot\|_1.
\end{equation}
This should hold with corrections for approximate designs, too, cf.~\cite{AmbainisEmerson}.

Obviously, as with Bosonic systems we are in infinite dimension, the 
dimension bound is \emph{a priori} not going to be useful. However, we can take 
inspiration from Corollary \ref{cor:W+lowerbound} and its proof, where 
we assume energy-bounded states, which we cut off at a finite photon number, 
restricting them thus to a finite-dimensional subspace. 

The more serious obstacle is that we would have to construct a Gaussian 
measurement, or a probabilistic mixture of Gaussian measurements, that 
approximates a $2$-design. 
But while the set of all Gaussian states has a locally compact symmetry group 
(symplectic and displacement transformations in phase space) that is consistent 
with a $2$-design, notorious normalisation and convergence issues 
prevent us from treating it as such \cite{Blume-KohoutTurner}. 

A different approach would be to analyse an even simpler measurement, 
which however must be tomographically complete. A nice candidate would 
be heterodyne detection on each mode, Eq. (\ref{eq:hetero}).

\section{Discussion}
\label{sec:conclusions}
By analysing the Wigner functions of Bosonic quantum states, we showed that 
there can be arbitrarily large gaps between the GOCC norm distance and the 
trace distance. In terms of the norm based on POVMs with positive Wigner 
functions, we could show that the separation necessarily requires
many modes, if we are in the regime of states with bounded energy per mode. 

Our results beg several questions, among them the following: 
first, is it possible to derandomise the construction of Theorem \ref{thm:GOCC-hiding}, 
in the sense that we would like to have concrete (not random) states with 
guaranteed separation of GOCC vs trace norm? 
Secondly, while our construction requires multiple modes, is it possible to 
have GOCC data hiding in a fixed number of modes, or even a single mode, 
at the expense of larger energy (cf.~Corollary \ref{cor:W+lowerbound})? 

Fortuitously, the recent work by Lami \cite{LL:new} goes some way 
towards addressing these questions: Indeed, \cite[Ex.~5]{LL:new}
shows two orthogonal Fock-diagonal states, called the
\emph{even and odd thermal states}, which while being at maximum 
possible trace distance, have arbitrarily small GOCC (and indeed W+) 
distance for sufficiently large energy (temperature). The
resulting upper bound \cite[Eq.~(23)]{LL:new} even compares well 
with our lower bound from Corollary \ref{cor:W+lowerbound}, when $m=1$.

The main difference to our scheme in Theorem \ref{thm:GOCC-hiding}
is that those even and odd thermal states are not Gaussian, or even 
mixtures of Gaussian states, in fact they have negative Wigner function,
indicating the difficulty in creating them. Instead our states, 
while undoubtedly complex (being multi-mode and requiring subtle 
arrangements of points in phase space, are simply uniform mixtures 
of coherent states, so in a certain sense they are easy to prepare
(an experimental implementation would be however still be challenging). 

As a matter of fact, this is best expressed in resource theoretic 
terms, noticing that GOCC actually can be defined as a class of 
quantum maps (to be precise: instruments), beyond our Definition 
\ref{defi:GOCC} of only GOCC measurements. This point of view is 
clearly evident in earlier references \cite{TakeokaSasaki}, even 
if it is not formalised. But recently, several attempts have been 
made to create fully-fledged resource theories of non-Gaussianity and of 
Wigner-negativity \cite{ZSS,TZ,AGPF}. While these works specifically 
focus on state transformations, and in particular the distillation of 
some form of ``pure'' non-Gaussian resource, our problem of the 
creation and discrimination of data hiding states are naturally 
phrased in the general resource theory. Indeed, in the framework 
of \cite{TZ,AGPF}, our GOCC measurements are free operations, and 
so are the state preparation of $\rho_0$ and $\rho_1$ from 
Theorem \ref{thm:GOCC-hiding}.
Thus, our results can be interpreted as contributions towards 
assessing the non-Gaussianity (Wigner negativity) of a measurement 
that distinguishes two states optimally. Here is the largest difference
to the cited recent papers, which formalise the resource character
of states, whereas our focus is on quantum operations. 
In that sense, Theorem \ref{thm:GOCC-hiding} (and equally 
\cite[Ex.~5]{LL:new}) provides a benchmark for the realisation 
of non-Gaussian quantum information processing, simply because 
optimal, or even decent discrimination of the states requires 
considerable abilities beyond the Gaussian (``linear'') realm.

\acknowledgments
The authors are grateful to Toni Ac\'{\i}n and Gael Sent\'{\i}s for prompting
the first formulation of the question treated in the present paper, during 
and after the doctoral defence of Gael, and in particular for sharing 
Ref.~\cite{TakeokaSasaki}. 
Thanks to John Calsamiglia and Ludovico Lami for asking many further 
questions which directed the present research, in particular about the 
GOCC Chernoff coefficient.
After the present work having been suspended for many years, we especially 
thank Ludovico Lami, whose keen interest in data hiding in general, and 
recent work \cite{LL:new} in particular, have eventually provided the 
motivation to finish and publish the present manuscript. 
%
%
Finally, we thank Prof. Luitpold Blumenduft for elucidating an optical 
phenomenon that bears a certain analogy to the phenomenon of Gaussian data hiding.

The authors' work was supported by the European Commission (STREP ``RAQUEL''), 
the ERC (Advanced Grant ``IRQUAT''),
the Spanish MINECO (grants FIS2008-01236, FIS2013-40627-P, 
FIS2016-86681-P and PID2019-107609GB-I00), with the support of FEDER funds,
and by the Generalitat de Catalunya, CIRIT projects 
2014-SGR-966 and 2017-SGR-1127.

\end{document}